\documentclass[11pt,letterpaper]{article} 
\usepackage{amsmath,amssymb,amsthm}
\usepackage{url}

\newtheorem{dfn}{Definition}
\newtheorem{fact}{Fact}
\newtheorem{prop}[fact]{Proposition}

\newtheorem{thm}{Theorem}
\newtheorem{cor}[thm]{Corollary}
\newtheorem{lemma}[thm]{Lemma}

\def\beq{\begin{equation}}
\def\eeq{\end{equation}}
\def\m#1{\mskip#1mu}
\def\n#1{\mskip-#1mu}

\def\p#1{\mathop{{\rm Pr}[\m1#1\m1]}}
\def\pr#1#2{\mathop{{\rm Pr}[\m1#1\m1|\m1#2\m1]}}
\def\H#1#2{{H}(#1\m1|\m1#2)}
\def\II#1#2{{I}(#1;#2)}

\def\IC{\mrm{IC}}
\def\D#1#2{D(#1|\n2|#2)}

\newcommand{\tup}[1]{({#1})}

\newcommand{\mbf}[1]{\mathbf{#1}}

\newcommand{\mrm}[1]{\ensuremath{\mathrm{#1}}}
\newcommand{\mcl}[1]{\mathcal{#1}}

\newcommand{\E}{\mathbf{E}}
\newcommand{\AND}{\ensuremath{\mathsf{AND}}}
\newcommand{\DISJ}{\ensuremath{\mathsf{DISJ}}}

\begin{document}

\title{Hellinger volume and number-on-the-forehead\\ communication complexity}
\author{
  Troy Lee\thanks{
    School of Physics and Mathematical Sciences, Nanyang Technological University and 
    Centre for Quantum Technologies.  Supported in part by the Singapore National Research 
    Foundation under NRF RF Award No. NRF-NRFF2013-13 and by a NSF postdoctoral fellowship
    while at Rutgers University.
    postdoctoral fellowship.
    Email:~{\tt troyjlee@gmail.com}.
  }
  \and Nikos Leonardos\thanks{
    {Department of Computer Science, Rutgers University, NJ, USA}
    {Supported in part by NSF under grant CCF 0832787.}
    Email:~{\tt nikos.leonardos@gmail.com}.
  }
  \and Michael Saks\thanks{
    {Mathematics Department, Rutgers University, NJ, USA}
    {Supported in part by NSF under grant CCF 0832787.}
    Email:~{\tt saks@math.rutgers.edu}.
  }
  \and Fengming Wang\thanks{
    {Department of Computer Science, Rutgers University, NJ, USA}
    {Supported in part by NSF under grants CCF 0830133, CCF 0832787, and CCF 1064785.}
    Email:~{\tt fengming@cs.rutgers.edu}.
  }
}
\maketitle

\begin{abstract}
Information-theoretic methods have proven to be a very powerful tool 
in communication complexity, in particular giving an elegant proof 
of the linear lower bound for the two-party disjointness function, and tight lower bounds 
on disjointness in the multi-party number-in-the-hand (NIH) model.  
In this paper, we study the applicability of information theoretic methods 
to the multi-party number-on-the-forehead model (NOF), where determining the 
complexity of disjointness remains an important open problem. 

There are two basic parts to the NIH disjointness lower bound: a direct sum 
theorem and a lower bound on the one-bit $\AND$ function using a beautiful 
connection between Hellinger distance and protocols revealed by 
Bar-Yossef, Jayram, Kumar \& Sivakumar \cite{byjks04}.
Inspired by this connection, we introduce the notion of Hellinger volume.  
We show that it lower bounds the
information cost of multi-party NOF protocols and provide a small toolbox that
allows one to manipulate several Hellinger volume terms and lower
bound a Hellinger volume when the distributions involved satisfy certain
conditions. In doing so, we prove a new upper bound on the difference between
the arithmetic mean and the geometric mean in terms of relative entropy.
We then apply these new tools to obtain a lower bound on the
informational complexity of the $\AND_k$ function in the NOF setting.  Finally, we discuss the 
difficulties of proving a direct sum theorem for information cost in the NOF model.  

\bigbreak\noindent{\bf Keywords:}
communication complexity, informational complexity, Hellinger volume,
number-on-the-forehead.
\end{abstract}

\section{Introduction}
One of the most important research areas in communication
complexity is proving lower
bounds in the multi-party number-on-the-forehead (NOF) model.   
The NOF model was introduced in \cite{cfl83}, where it was used to
prove lower bounds for branching programs. Subsequent papers revealed
connections of this model to circuit complexity \cite{bt94,hg90,nisan94,nw91} and proof
complexity \cite{bps05}.  In particular, an explicit function which requires
super-polylogarithmic complexity in the NOF model with polylogarithmically many
players would give an explicit function outside of the circuit complexity class
$\mathsf{ACC}^0$.

Essentially all lower bounds on the general NOF model have been shown using the
discrepancy method following \cite{bns92}. This method has been able to show
lower bounds of $\Omega(n/2^k)$ for explicit functions
\cite{bns92,ct93,raz00,fg05}. For the disjointness function, the plain discrepancy method shows 
poor bounds and a more sophisticated application of  
discrepancy is needed known as the generalized discrepancy method \cite{Kla07, Raz03, LS07, She08}.  The generalized discrepancy method was initially
used to show lower bounds of the form $n^{1/k}/2^{2^k}$ \cite{ls09,CA08} and
$2^{\Omega(\sqrt{\log n}/\sqrt{k})-k}$ \cite{bn09} on the $k$-player NOF complexity of disjointness. 
Recent work of Sherstov in \cite{sherstov2012} and
\cite{sherstov2013} improved the lower bounds to
$\Omega((n/4^k)^{1/4})$ and $\Omega(\sqrt n/2^kk)$, respectively.
A very recent paper of Rao and Yehudayoff \cite{ry2014} gives a simplified
proof of the latter lower bound and also gives a nearly tight $\Omega(n/4^k)$
lower bound for deterministic protocols.  An upper bound of $O(\log^2(n) + k^2n/2^k)$ for the disjointness function 
follows from a beautiful protocol of Grolmusz \cite{gro94}. 

In this paper we are interested in
how information-theoretic methods might be applied to the NOF model.  Information-theoretic 
methods have been very successful in the number-in-the-hand (NIH) multi-party model, 
in particular giving tight lower bounds on the disjointness function.  The first
use of information theory in communication complexity lower bounds can be traced
to \cite{abl96}.  In \cite{cswy01} the notions of information cost and
informational complexity were defined explicitly. Building on their work, a very
elegant information-theoretic framework for proving lower bounds in NIH communication 
complexity was established in \cite{byjks04}. 

In \cite{byjks04} a proof of the linear lower bound for two-party disjointness
was given. The proof has two main stages. In the first stage, a direct-sum
theorem for informational complexity is shown, which says that the
informational complexity of disjointness,
$\DISJ_{n,2}(x,y)=\bigvee_{j=1}^n\AND_2(x_j,y_j)$, is lower bounded by $n$ times
the informational complexity of the binary $\AND_2$ function. 
Although it is not known how to prove such a direct-sum theorem directly for the
classical randomized complexity, Bar-Yossef et al.\ prove it for the
informational complexity with respect to a suitable distribution. A crucial
property of the distribution is that it is over the zeroes of disjointness. At
this point we should point out a remarkable characteristic of the method: even
though the information cost of a protocol is analyzed with respect to a
distribution over zeroes only, the protocol is required to be correct over all
inputs. This requirement is essential
in the second stage, where
a constant lower bound
is proved on the informational complexity of $\AND_2$. 
This is achieved using properties of the Hellinger distance for distributions. 
Bar-Yossef et al.\ reveal a beautiful connection between Hellinger distance and
NIH communication protocols. (More properties of Hellinger distance relative to
the NIH model have been established in \cite{j09}.)

In this work we provide tools for accomplishing the second stage in the NOF
model. We introduce the notion of Hellinger volume of $m\ge2$ distributions and show
that it can be useful for proving lower bounds on informational complexity in
the NOF model, just as Hellinger distance is useful in the NIH model.
However, as we point out in the last section, there are fundamental
difficulties in proving a direct-sum theorem for informational complexity in the
NOF model. Nevertheless, we believe that Hellinger volume and the related tools
we prove, could be useful in an information-theoretic attack on NOF complexity.

A version of this paper was submitted to a journal in 2011, but the refereeing
process has been long delayed. In the meantime there has been some overlapping
independent work by Beame, Hopkins, Hrube\v{s} and Rashtchian \cite{Beame},
including lower bounds for the information complexity of the AND function
similar to those we give in Section~\ref{section5} but for restricted settings,
and 0-information protocols in the ``randomness on the forehead'' model, of the
\enlargethispage\baselineskip
type we give in Section~\ref{section6} but in a more general setting.

\section{Preliminaries and notation}
\paragraph{Hellinger volume}
We introduce the notion of Hellinger volume of $m$ distributions. In the next
section we show that it has properties similar in flavor to the ones of
Hellinger distance.
\begin{dfn}
	The $m$-dimensional Hellinger volume of distributions $p_1,\dots,p_m$ over
	$\Omega$ is
	\[h_m(p_1,\dots,p_m)=
		1-\sum_{\omega\in\Omega}\sqrt[m]{p_1(\omega)\cdots p_m(\omega)}.\]
\end{dfn}
Notice that $h_2(p_1,p_2)$ in the case $m=2$ is the square of the Hellinger
distance between distributions $p_1$ and $p_2$.

The following fact follows from the arithmetic-geometric mean inequality.
\begin{fact}\label{fact.nonneg}
	For any distributions $p_1,\dots,p_m$ over\/ $\Omega$,
	$h_m(p_1,\dots,p_m)\ge0$.
\end{fact}

\paragraph{Random variables and distributions}
We consider discrete probability spaces $(\Omega,\zeta)$, where
$\Omega$ is a finite set and $\zeta$ is a nonnegative valued function
on $\Omega$ summing to 1.
If $(\Omega_1,\zeta_1),\ldots,(\Omega_n,\zeta_n)$
are such spaces, their product is the space $(\Lambda,\nu)$,
where $\Lambda=\Omega_1\times\cdots\times\Omega_n$ is the
Cartesian product of sets, and for $\omega=(\omega_1,\ldots,\omega_n)\in\Lambda$,
$\nu(\omega)=\prod_{j=1}^n\zeta_j(\omega_j)$.  In the case that all of
the $(\Omega_i,\zeta_i)$ are equal to a common space $(\Omega,\zeta)$ we write
$\Lambda=\Omega^n$ and $\nu=\zeta^n$.

We use uppercase for random variables, as in $\mbf{Z},D$, and write in bold
those that represent vectors of random variables. For a variable $X$ with range
$\mcl{X}$ that is distributed according to a probability distribution $\mu$,
i.e.\ $\Pr[X=x]=\mu(x)$, we write $X\sim\mu$. If $X$ is uniformly distributed in
$\mcl{X}$, we write $X\in_R\mcl{X}$.

\paragraph{Information theory}
Let $X,Y,Z$ be random variables on a common probability space, taking on values,
respectively, from finite sets $\mcl{X},\mcl{Y},\mcl{Z}$.
Let $A$ be any event.
The {\em entropy} of $X$, the {\em conditional entropy of $X$ given
$A$}, and the {\em conditional entropy of $X$ given $Y$} are respectively 
(we use $\log$ for $\log_2$)
\begin{align*}
	H(X)&=-\sum_{x\in\mcl{X}}\Pr[X=x]\cdot\log\Pr[X=x],\\
	\H XA&=-\sum_{x\in\mcl{X}}\Pr[X=x\,|\,A]\cdot\log\Pr[X=x\,|\,A],\\
	\H XY&=\sum_{y\in\mcl{X}}\Pr[Y=y]\cdot\H X{Y=y}.
\end{align*}
We will need the following facts about the entropy. (See
\cite[Chapter~2]{ct06}, for proofs and more details.)
\begin{prop}\label{entropy}
	Let $X,Y,Z$ be random variables.
	\begin{enumerate}
		\item\label{entropy:1} 
			${H}(X)\ge{H}(X\,|\,Y)\ge0$.
		\item\label{entropy:2} 
			If $\mcl{X}$ is the range of $X$, then
			${H}(X)\le\log|\mcl{X}|$.
		\item\label{entropy:3} 
			${H}(X,Y)\le{H}(X)+{H}(Y)$ with equality if and
			only if $X$ and $Y$ are independent. This holds for conditional entropy
			as well.  $\H{X,Y}Z\le\H XZ+\H YZ$ with equality if and only if $X$ and
			$Y$ are independent given $Z$.
	\end{enumerate}
\end{prop}
The {\em relative entropy} or {\em divergence} of distributions $P$ and $Q$ over
$\Omega$ is
\[\D PQ=\sum_{x\in\Omega}P(x)\log{P(x)\over Q(x)}.\]

The {\em mutual information} between $X$ and $Y$ is
\[\II XY=H(X)-\H XY=H(Y)-\H YX.\]

\paragraph{Notation}
We write $[n]=\{1,2,\dots,n\}$.  For a sequence $\tup{a_1,\dots,a_n}$
we let, for $j\in[n]$, $a_{<j}=\tup{a_1,\dots,a_{j-1}}$, and
$a^{-j}=(a_1,\dots,a_{j-1},a_{j+1},\dots,a_k)$.
We will denote subsets of $\{0,1\}^k$ as follows: $I=\{0,1\}^k$;
for $j\in[k]$, $I_j$ is the set of points in $I$ such that the $j$-th coordinate
is set to zero, i.e.\ $I_j=\{z\in I\mid z_j=0\}$; $I_{OZ}$ (resp.\ $I_{EZ}$) is
the set of points in $I$ with an odd (resp.\ even) number of zeros.

\paragraph{Communication complexity}
In this work we will be dealing with the multi-party private-coin randomized
number-on-the-forehead communication model, introduced by \cite{cfl83}. There are
$k$ players, numbered $1,\dots,k$, trying to compute a function
$f:\mathcal{Z}\to\{0,1\}$, where
$\mathcal{Z}=\mathcal{Z}_1\times\cdots\times\mathcal{Z}_k$.
On input $z\in\mathcal{Z}$, player $j$ receives input $z_j$ (conceptually,
placed on his forehead), but he has access only to $z^{-j}$.
They wish to determine $f(z)$, by broadcasting messages according to a protocol
$\Pi$. Let the random variable $\Pi(z)$ denote the transcript of the
communication on input $z$ (where the probability is over the random coins of
the players) and $\Pi_\text{out}(z)$ the outcome of the protocol.
We call $\Pi$ a {\em $\delta$-error protocol for $f$} if, for all $z$,
$\Pr[\Pi_\text{out}(z)=f(z)]\ge1-\delta$. The {\em communication cost of
$\Pi$} is $\max|\Pi(z)|$, where the maximum is over all inputs
$z$ and over all coin tosses of the players. The {\em $\delta$-error randomized
communication complexity of $f$}, denoted $R_\delta(f)$, is the cost of the best
$\delta$-error protocol for $f$. (See \cite{kn06} for more details.)

\paragraph{Communication complexity lower bounds via information theory}
The informational complexity paradigm, introduced by
\cite{cswy01}, and used in \cite{ss02,byjks02,cks03,byjks04,jks03},
provides a way to prove lower bounds on communication
complexity via information theory.
We are given a $k$-party function $f$ and we want to show that any
$\delta$-error randomized NOF protocol $\Pi$ for $f$ requires high
communication.
We introduce a probability
distribution over the inputs to the players. We then analyze the 
behavior of $\Pi$ when run on inputs chosen randomly according to the
distribution.  The informational complexity is the mutual information
of the string of communicated bits (the {\em transcript} of $\Pi$)
with the inputs, and provides a lower bound on
the amount of communication.

More precisely,
let $\Omega=(\Omega,\zeta)$ be a probability space over which are defined
random variables $\mbf{Z}=\tup{Z_1,\ldots,Z_k}$ representing the inputs.
The {\em information cost} of a protocol $\Pi$ with respect to $\zeta$
is defined to be $\II{\mbf{Z}}{\Pi(\mbf{Z})}$, where
$\Pi(\mbf{Z})$ is a random variable following the distribution of the
communication transcripts when the protocol $\Pi$ runs on input
$\mbf{Z}\sim\zeta$. The {\em $\delta$-error informational complexity} of
$f$ with respect to $\zeta$, denoted $\IC_{\zeta,\delta}(f)$, is
$\min_\Pi\II{\mbf{Z}}{\Pi(\mbf{Z})}$, where the minimum is over all
$\delta$-error randomized NOF protocols for $f$.
The relevance of informational complexity comes from the following proposition.
\begin{prop}
	$R_\delta(f)\ge\IC_{\zeta,\delta}(f)$.
\end{prop}
\begin{proof}
	For any protocol $\Pi$,
	$\IC_{\zeta,\delta}(f)\le\II{\mbf{X,Y}}{\Pi(\mbf{X,Y})}=
	  H(\Pi(\mbf{X,Y}))-\H{\Pi(\mbf{X,Y})}{\mbf{X,Y}}$.
			Applying in turn parts (1) and (2) of Proposition~\ref{entropy} gives
	$\IC_{\zeta,\delta}(f)\le{H}({\Pi(\mbf{X,Y})})\le R_\delta(f)$. 
\end{proof}

For a collection of distributions $\eta=\{\zeta_1,\dots,\zeta_k\}$, we define
the {\em $\delta$-error informational complexity} of $f$ with respect to $\eta$,
denoted $\IC_{\eta,\delta}(f)$, to be $\E_j[\IC_{\zeta_j,\delta}(f)]$, where $j$
is uniformly distributed over $[k]$.

\paragraph{Remark} This definition of information cost as an average, is
equivalent to the (standard) conditional information cost. We choose this
definition, because we think it makes the exposition cleaner.

\section{An upper bound on the difference between the arithmetic and geometric mean.}
For a nonnegative real sequence $\alpha=(\alpha_1,\dots,\alpha_m)$, let
$A(\alpha)$ and $G(\alpha)$ denote its arithmetic and geometric mean
respectively. That is
\[A(\alpha)=\frac{1}{m} \sum\alpha_j \quad\text{and}\quad G(\alpha)=
	\sqrt[m]{\prod\alpha_j}.\]	

\begin{thm}\label{thm.diff}
	For any distribution $p$ over $[m]$,
	\[A(p)-G(p)\le\ln2\cdot\D pu,\]
	where $u$ is the uniform distribution over $[m]$.
\end{thm}
\begin{proof}
	Let $x_j=mp(j)$, $x=\tup{x_1,\ldots,x_n}$, and define
	\[f(x)=\sum x_j\ln{x_j}\m3+\sqrt[m]{\prod x_j}.\]
	Theorem~\ref{thm.diff} is equivalent to showing that, for $x_1,\dots,x_n\ge0$,
	if $\sum x_j=m$, then $f(x)\ge1.$

	We proceed using Lagrange multipliers. 
	We first need to check that $f(x)\ge1$ when $x$ is on the boundary, i.e.\
	$x_j=0$ for some $j\in[n]$. Without loss of generality, assume $x_1=0$. 
	By the convexity of $t\ln t$, the minimum is attained when
	$x_2=\dots=x_m=m/(m-1)$. Thus,
	\[f(x)\ge(m-1)\frac{m}{m-1}\ln\frac{m}{m-1}>
		m\biggl(1-\frac{m-1}{m}\biggr)=1.\]
	According to \cite[Theorem on page 300]{luen2003}, it suffices to show that $f(x)\ge1$ for
	any $x$ that satisfies the following system of equations.
	\[{\partial f}/{\partial x_j}=1+\ln x_j+\sigma/(mx_j)=\lambda,\quad
		\text{for $j\in[m]$},\tag{$L$}\]
	where $\sigma=\sqrt[m]{x_1\cdots x_m}\ne0$.
	Without loss of generality, since $\sum x_j=m$, we may assume $x_m\le1$. The
	system $(L)$ implies
	\begin{gather*}
		\sum_{j=1}^{m-1}x_j(\partial f/\partial x_j)=
			m-x_m+\sum_{j=1}^{m-1}x_j\ln x_j+\sigma(m-1)/m=\lambda(m-x_m),\\
		(m-1)x_m(\partial f/\partial x_m)=(m-1)(x_m+x_m\ln x_m+\sigma/m)
			=(m-1)\lambda x_m.
	\end{gather*}
	Subtracting the second from the first we get
	\[\sum_{j=1}^{m-1}x_j\ln x_j-(m-1)x_m\ln x_m=m(\lambda-1)(1-x_m).\]
	We also have
	\[\sum x_j(\partial f/\partial x_j)=m+f(x)=m\lambda.\]
	Suppose $x=(x_1,\dots,x_m)$ satisfies the system $(L)$. Since $x_m\le1$,
	we have $x_m\ln x_m\le0$, and using the last two equations we have
	\[f(x)=m(\lambda-1)\ge\frac{\sum_{j=1}^{m-1}x_j\ln x_j}{1-x_m}
		\ge\frac{\sum_{j=1}^{m-1}x_j(1-1/x_j)}{1-x_m}=1.\]
	This completes the proof.
\end{proof}

\begin{cor}
	For any nonnegative real sequence $\alpha=(\alpha_1,\dots,\alpha_m)$,
	\[A(\alpha)-G(\alpha)\le\sum\alpha_j\ln\frac{\alpha_j}{A(\alpha)}.\]
\end{cor}
\begin{proof}
	Apply Theorem~\ref{thm.diff} with $p(j)=\alpha_j\big/\n2\sum_j\alpha_j$.
\end{proof}
\paragraph{Remark} Let $\widehat\alpha$ to be a normalized version of $\alpha$,
with $\widehat\alpha_j=\alpha_j\big/\n2\sum\alpha_j$. Let also $u$ denote the uniform
distribution on $[m]$. Then, the right-hand side takes the form
$\sum\alpha_j\ln(m\widehat\alpha_j)=mA(\alpha)\sum\widehat\alpha_j\ln(\widehat\alpha_j/u_j)$,
and the above inequality becomes
	\[{A(\alpha)-G(\alpha)\over A(\alpha)}\le m\ln2\cdot\D{\widehat\alpha}{u}.\]

\section{Properties of Hellinger volume}
\paragraph{Hellinger volume lower bounds mutual information}
The next lemma shows that Hellinger volume can be used to lower bound
mutual information.
\begin{lemma}\label{lemma.mutual2hel}
	Consider random variables $Z\in_R[m]$, $\Phi(Z)\in\Omega$, and distributions
	$\Phi_z$, for $z\in[m]$, over\/ $\Omega$. Suppose that given $Z=z$, the
	distribution of $\Phi(Z)$ is $\Phi_z$. Then
	\[\II Z{\Phi({Z})}\ge\frac{h_m(\Phi_1,\dots,\Phi_m)}{m\ln2}.\]
\end{lemma}
\begin{proof}
	The left-hand side can be expressed as follows (see \cite[page 20]{ct06}),
	\begin{align*}
		\II Z{\Phi({Z})}&=\sum_{j,\omega}\p{Z=j}\cdot \pr{\Phi(Z)=\omega}{Z=j}
			\cdot\log\frac{\pr{\Phi(Z)=\omega}{Z=j}}{\p{\Phi(Z)=\omega}}\\
		&=\sum_{j,\omega}\frac{1}{m}\Phi_j(\omega)
			\log\frac{\Phi_j(\omega)}{\frac{1}{m}\sum_j\Phi_j(\omega)},
	\end{align*}
	and the right-hand side
	\[h_m(\Phi_1,\dots,\Phi_m)=\sum_\omega\biggl(
		\frac{1}{m}\sum_j\Phi_j(\omega)-
			\Bigl(\prod_j\Phi_j(\omega)\Bigr)^\frac{1}{m}\biggr).\]
	It suffices to show that for each $\omega\in\Omega$,
	\[\sum_{j}\frac{1}{m}\Phi_j(\omega)
		\log\frac{\Phi_j(\omega)}{\frac{1}{m}\sum_j\Phi_j(\omega)}\ge
		\frac{1}{m\ln2}\biggl(\frac{1}{m}\sum_j\Phi_j(\omega)-
			\Bigl(\prod_j\Phi_j(\omega)\Bigr)^\frac{1}{m}\biggr).\]
	Let $s=\sum_j\Phi_j(\omega)$, and $\rho(j)=\Phi_j(\omega)/s$, for
	$j\in[m]$; thus, for all $j$, $\rho(j)\in[0,1]$, and $\sum_j\rho(j)=1$. Under
	this renaming of variables, the left-hand side becomes
	$\ln2\cdot\frac{s}{m}\sum_j\rho(j)\log(m\rho(j))$ and the right one
	$\frac{s}{m}\cdot({1\over m}-\root m\of{\prod\rho(j)})$.
	Thus, we need to show
	\[\ln2\cdot\sum_j\rho(j)\log(m\rho(j))\ge\frac{1}{m}
		-\Bigl(\prod_j\rho(j)\Bigr)^\frac{1}{m}.\]
	Observe that the left-hand side is $\ln2\cdot\D\rho u$, and the 
	inequality holds by Theorem~\ref{thm.diff}.
\end{proof}

\paragraph{Symmetric-difference lemma}
Let $P=\{P_z\}_{z\in Z}$ be a collection of distributions over a common space
$\Omega$. For $A\subseteq Z$, the {\em Hellinger volume of $A$ with respect to
$P$}, denoted by $\psi(P;A)$, is
\[\psi(A;P)=1-\sum_{\omega\in\Omega}\Bigl({\prod_{z\in A}P_z(\omega)}\Bigr)^{{1}/{|A|}}.\]
The collection $P$ will be understood from the context and we'll say that the
Hellinger volume of $A$ is $\psi(A)$. Note that, from Fact~\ref{fact.nonneg},
$\psi(A;P)\ge0$.

The following lemma can be seen as an analog to the weak triangle inequality
that is satisfied by the square of the Hellinger distance.
\begin{lemma}[Symmetric-difference lemma]\label{lemma.sym_dif}
	If $A,B$ satisfy $|A|=|B|=|A\Delta B|$, where
	$A\Delta B =(A\setminus B)\cup(B\setminus A)$.  Then
	\[\psi(A)+\psi(B)\ge\frac{1}{2}\cdot\psi(A\Delta B).\]
\end{lemma}
\begin{proof}
By our hypothesis, it follows that $A \setminus B$, $B \setminus A$ and $A \cap B$ all have
size $|A|/2$.
Define $u,v,w$ to be the vectors in $\mathbb{R}^{\Omega}$ defined by
\begin{align*}
	u(\omega)&=\Bigl(\n8{\prod_{z\in A\setminus B}\n{10}P_z(\omega)}\Bigr)^{1/|A|},\\
	v(\omega)&=\Bigl(\n8\prod_{z\in B\setminus A}\n{10}P_z(\omega)\Bigr)^{1/|A|},\\
	w(\omega)&=\Bigl(\n8\prod_{z\in A\cap B}\n{10}P_z(\omega)\Bigr)^{1/|A|}.
\end{align*}
By the definition of Hellinger volume,
\begin{align*}
	\psi(A)&=1-u\cdot w,\\
	\psi(B)&=1-v\cdot w,\\
	\psi(A\Delta B)&=1-u\cdot v.
\end{align*}
Thus the desired inequality is
\[2-(u+v)\cdot w\ge(1- u \cdot v)/2,\]
which is equivalent to
\begin{equation}\label{vector ineq 1}
	3+u\cdot v\ge2(u+v)\cdot w.
\end{equation}
Since
\begin{align*}
	\psi(A\setminus B)&=1-u \cdot u,\\
	\psi(B\setminus A)&=1-v \cdot v,\\
	\psi(A\cap B)&=1-w \cdot w,
\end{align*}
it follows that $\|u\|,\|v\|$ and $\|w\|$ are all at most 1.
Thus $2(u+v)\cdot w\le2\|u+v\|$, and so (\ref{vector ineq 1})
follows from
\[3+u\cdot v\ge2\|u+v\|.\]
Squaring both sides, it suffices to show
\[9+ 6u \cdot v+(u \cdot v)^2  \ge4(\|u\|^2+\|v\|^2+2 u \cdot v)\]
Using the fact that $\|u\|\le1$ and $\|v\|\le1$ this
reduces to
\[(1-u\cdot v)^2\ge0,\]
which holds for all $u,v$.
\end{proof}

Let $s_l,s_r$ be two disjoint subsets of $[k]$. Let $I_l \subseteq I$ (resp.,
$I_r$) be the set of strings with odd number of zeros in the coordinates indexed
by $s_l$ (resp., $s_r$). Let $s_p = s_l \cup s_r$ and
$I_p = I_l \Delta I_r$. It is not hard to see that $I_p$ is the set of strings
with odd number of zeros in the coordinates indexed by $s_p$.  By the
symmetric-difference lemma,
\begin{equation}\label{applicationOfSym}
\psi(I_l) + \psi(I_r) \ge \frac{\psi(I_p)}{2}.
\end{equation}

For each $j\in[k]$, let $I_j\subseteq I$ be the set of strings where the $j$-th
coordinate is set to zero.  Applying the above observation inductively, we can
obtain the following lemma.

\begin{lemma}\label{applicationOfSymGeneral}
	Let $s \subseteq [k]$ be an arbitrary non-empty set and let $I_s \subseteq I$
	be the set of strings with odd number of zeros in the coordinates indexed by
	$s$. Then,
	\[\sum_{j \in s}\psi(I_j)\ge\frac{\psi(I_s)}{2^{\lceil\log |s| \rceil}}.\]
\end{lemma}
\begin{proof}
We prove the claim via induction on the size of $s$. If $s$ is a singleton set,
it trivially holds. Otherwise, assume that for any subset of $[k]$ of size less
than $|s|$, the claim is true.

Partition $s$ into two non-empty subsets $s_l,s_r$ with the property that
$|s_l|=\lceil|s|/2\rceil$ and $|s_l|=\lfloor|s|/2\rfloor$.
Then $\lceil\log|s|\rceil=1+\max\{\lceil\log|s_l|\rceil,\lceil\log|s_r|\rceil\}$.
By the inductive hypothesis,
\[\sum_{j \in s_l}\psi(I_{s_l})\ge\frac{\psi(I_{s_l})}{2^{\lceil\log|s_l|\rceil}}
	\quad\hbox{and}\quad
	\sum_{j\in s_r}\psi(I_{s_r})\ge\frac{\psi(I_{s_r})}{2^{\lceil\log|s_r|\rceil}}.\]
Thus,
\begin{align*}
	\sum_{j \in s}\psi(I_{s_l})
    &=\sum_{j \in s_l}\psi(I_{s_l})+\sum_{j \in s_r}\psi(I_{s_r})\\
		&\ge\frac{\psi(I_{s_l})}{2^{\lceil\log|s_l|}}+
			\frac{\psi(I_{s_r})}{2^{\lceil\log |s_r| \rceil}}
			&\hbox{by the Inductive Hypothesis,}\\
   	&\ge\frac{1}{2^{\lceil\log|s|\rceil-1}}[\psi(I_{s_l})+\psi(I_{s_r})]
			&\hbox{by the choice of $s_l$ and $s_r$,}\\
   	&\ge\frac{1}{2^{\lceil\log|s|\rceil}}\psi(I_s)
			&\hbox{by Equation (\ref{applicationOfSym}).}
\end{align*}
\end{proof}

Let $I_{OZ}\subseteq I$ be the set of strings which have odd number of zeros.
The next corollary is an immediate consequence of Lemma
\ref{applicationOfSymGeneral} when $s = [k]$.

\begin{lemma}\label{lemma.tree}
	\[\sum_{j=1}^k\psi(I_j)\ge\frac{\psi(I_{OZ})}{2^{\lceil\log k\rceil}}.\]
\end{lemma}

\paragraph{NOF communication complexity and Hellinger volume}
It was shown in \cite{byjks04} that the distribution of transcripts of a
two-party protocol on a fixed input is a product distribution. The same is true
for a multi-party NOF protocol.

\begin{lemma}
	Let $\Pi$ be a $k$-player NOF communication protocol with input set
	${\cal Z}={\cal Z}_1\times\cdots\times{\cal Z}_k$ and let\/ $\Omega$ be the
	set of possible transcripts. For each $j\in[k]$, there is
	a mapping $q_j:\Omega\times{\cal Z}^{-j}\to\mathbb{R}$, such that for every
	$z=(z_1,\dots,z_k)\in{\cal Z}\!$ and $\omega\in\Omega$,
	\[\p{\Pi(z)=\omega}=\prod_{j=1}^kq_j(\omega;z^{-j}).\]
\end{lemma}
\begin{proof}
	Suppose $|\Pi(z)|\le l$. For $i=1,\dots,l$, let $\Pi_i(z)$ denote the $i$-th bit
	sent in an execution of the protocol. Let $\sigma_i\in[k]$ denote the player
	that sent the $i$-th bit. Then
	\begin{align*}
		\p{\Pi(z)=\omega}&=\p{\Pi_1(z)=\omega_1,\dots,\Pi_l(z)=\omega_l}\\
		&=\prod_{i=1}^l\pr{\Pi_i(z)=\omega_i}{\Pi_{<i}(z)=\omega_{<i}},\\
		&=\prod_{i=1}^l\p{\Pi_i(z^{-\sigma_i};\omega_{<i})=\omega_i},
	\end{align*}
	because every bit send by player $j$ depends only on $z^{-j}$ and the
	transcript up to that point. We set
	\[q_j(\omega;z^{-j})=\n6
		\prod_{i:\sigma_i=j}\n6\p{\Pi_i(z^{-j};\omega_{<i})=\omega_i}\]
	to obtain the expression of the lemma.
\end{proof}

As a corollary, we have the following cut-and-paste property for Hellinger
volume.
\begin{lemma}\label{lemma.cutandpaste}
	Let $I_{OZ}\subseteq I$ be the set of inputs which have odd number of zeros,
	and let $I_{EZ}=I\setminus I_{OZ}$. Then \[\psi(I_{OZ})=\psi(I_{EZ}).\]
\end{lemma}
\begin{proof}
Using the expression of the previous lemma, we have that for any
$\omega\in\Omega$,
\begin{align*}
	\prod_{v\in I_{OZ}}\n8P_v(\omega)
	\m6=\prod_{v\in I_{OZ}}\prod_{j=1}^kq_j(\omega;v^{-j})
	\m6=\prod_{u\in I_{EZ}}\prod_{j=1}^kq_j(\omega;u^{-j})
	\m6=\prod_{u\in I_{EZ}}\n8P_u(\omega).
\end{align*}
The middle equality holds, because for each $j\in[k]$ and $v\in I_{OZ}$ there is
a unique $u\in I_{EZ}$ such that $v^{-j}=u^{-j}$.
\end{proof}

\paragraph{Lower bounding Hellinger volume}
Eventually, we will need to provide a lower bound for the Hellinger volume of
several distributions over protocol transcripts. In the two-party case, one
lower bounds the Hellinger distance between the distribution of the transcripts
on an accepting input and the distribution of the transcripts on a rejecting
input. The following lemma will allow for similar conclusions in the multi-party
case.
\begin{lemma}\label{lemma.lb}
	Let $A \subseteq I$ be of size $t\ge2$. Suppose there is an event
	$T\subseteq\Omega$, a constant $0\le\delta\le1$ and an element $v$ in $A$ such
	that\/ $P_v(T)\ge1-\delta$ and that for all $u\in A$
	with $u\neq v$, $P_u(T)\le\delta$. Then
	\[\psi(A)\ge\bigl(2-4\sqrt{\delta(1-\delta)}\bigr)\cdot\frac{1}{t}.\]
\end{lemma}
\begin{proof}
We need to show
\[1-\sum_{\omega\in\Omega}\prod_{u\in A}P_u(\omega)^{\frac{1}{t}}\ge
	\bigl(2-4\sqrt{\delta(1-\delta)}\bigr)\cdot\frac{1}{t}.\]
Let $a=P_v(T)=\sum_{\omega\in T} P_v(\omega)$ and
$b=\sum_{\omega \in T} \tfrac{1}{t-1} \sum_{u \neq v} P_u(\omega)$.  Notice that by assumption
$a \ge 1-\delta$ and $b \le \delta$.

Recall H\"older's inequality: for
any nonnegative $x_k$, $y_k$, $k\in m$, 
\[\sum_{k=1}^mx_ky_k\le\Bigl(\sum_{k=1}^mx_k^t\Bigr)^{1\over t}
	\Bigl(\sum_{k=1}^my_k^\frac{t}{t-1}\Bigr)^{t-1\over t}.\]
We first treat the sum over $\omega\in T$. 
\begin{align*}
	\sum_{\omega\in T}\prod_{u\in A}P_u(\omega)^{\frac{1}{t}}
	&=\sum_{\omega\in T}P_v(\omega)^{\frac{1}{t}}\prod_{u\ne v}P_u(\omega)^{\frac{1}{t}}\\
	&\le\Bigl(\sum_{\omega\in T}P_v(\omega)\Bigr)^{\frac{1}{t}}
	\Bigl(\sum_{\omega\in T}\prod_{u\ne v}P_u(\omega)^{\frac{1}{t-1}}\Bigr)^{\frac{t-1}{t}}\\
	&\le\Bigl(\sum_{\omega\in T}P_v(\omega)\Bigr)^{\frac{1}{t}}
	\Bigl(\sum_{\omega\in T}\frac{1}{t-1}\sum_{u\ne v}P_u(\omega)\Bigr)^{\frac{t-1}{t}}\\
	&=a^{\frac{1}{t}}b^{\frac{t-1}{t}},
\end{align*}
where we first used H\"older's inequality and then the arithmetic-geometric
mean inequality.
We do the same steps for the sum over $\omega \not \in T$ to find
\[\sum_{\omega\not\in T}\prod_{u\in A}P_u(\omega)^{\frac{1}{t}}
	\le(1-a)^{\frac{1}{t}}(1-b)^{\frac{t-1}{t}}.\]
Hence,
\[\sum_{\omega\in\Omega}\prod_{u\in A}P_u(\omega)^{\frac{1}{t}}\le
	a^{\frac{1}{t}}b^{\frac{t-1}{t}}+(1-a)^{\frac{1}{t}}(1-b)^{\frac{t-1}{t}}.\]

Let $g(a,b,x) = a^{x}b^{1-x}+(1-a)^{x}(1-b)^{1-x}$. We will show that under the
constraints $a\ge1-\delta$ and $b\le\delta$ where $\delta<1/2$, for
any fixed $0\le x\le1/2$, $g(a,b,x)$ is maximized for $a=1-\delta$
and $b=\delta$. The partial derivatives for $g(a,b,x)$ with respect to $a$ and
$b$ are

\[g_a(a,b,x)=x[a^{x-1}b^{1-x}-(1-a)^{x-1}(1-b)^{1-x}]
	=x\Bigl[\Bigl(\frac{b}{a}\Bigr)^{1-x}-\Bigl(\frac{1-b}{1-a}\Bigr)^{1-x}\Bigr]\]
\[g_b(a,b,x)=(1-x)[a^{x}b^{-x}-(1-a)^{x}(1-b)^{-x}]
	=(1-x)\Bigl[\Bigl(\frac{b}{a}\Bigr)^{-x}-\Bigl(\frac{1-b}{1-a}\Bigr)^{-x}\Bigr]\]
Under our constraints, $\frac{b}{a} < 1 < \frac{1-b}{1-a}$, $1-x > 0$ and
$-x \le 0$, thus, $g_a(a,b,x)<0$ and $g_b(a,b,x)\ge0$ for any such $a$, $b$, and
$x$.  This implies that for any fixed $b$, $g(a,b,x)$ is maximized when
$a=1-\delta$ and similarly for any fixed $a$, $g(a,b,x)$ is maximized when 
$b=\delta$. Therefore, for all $a$, $b$, and $0\le x\le1$, $g(a,b,x)\le
g(1-\delta,\delta,x)$.

For $0\le x\le1/2$, let
\[f(\delta,x)=1-g(1-\delta,\delta,x)=
	1-(1-\delta)^x\delta^{1-x}-\delta^x(1-\delta)^{1-x}.\] 
Since $f(\delta,x)$ is convex for any constant $0\le\delta\le1$, 
\[f(\delta,x)\ge{f(\delta,1/2)-f(\delta,0)\over1/2-0}\cdot x=
	2\bigl(1-2\sqrt{\delta(1-\delta)}\bigr)\cdot x.\]
\end{proof}

\section{An application}\label{section5}
In this section we show how to derive a lower bound for the informational
complexity of the $\AND_k$ function. Define a collection of distributions
$\eta=\{\zeta_1,\dots,\zeta_k\}$, where, for each $j\in[k]$, $\zeta_j$ is the
uniform distribution over $I_j$.  Recall that $I_j\subseteq I=\{0,1\}^k$ for $j\in[k]$ is the set of 
$k$-bitstrings whose $j$-th bit is 0. 
We prove the following lower bound on the $\delta$-error informational
complexity of $\AND_k$ with respect to $\eta$.

\medskip\noindent
{\bf Remark.} The choice of the collection $\eta$ is not arbitrary, but is
suggested by the way the direct-sum theorem for informational complexity is
proved in \cite{byjks04} for the two-party setting. In particular, two
properties of $\eta$ seem crucial for such a purpose. First, for each $j\in[k]$,
$\zeta_j$ is a distribution with support only on the zeroes of $\AND_k$. 
Second, under any $\zeta_j$, the input of each player is independent of any
other input.

\begin{thm}
	\[\mathrm{IC}_{\eta,\delta}(\AND_k)\ge\log e\cdot
		\bigl(1-2\sqrt{\delta(1-\delta)}\bigr)\cdot\frac{1}{k^2\m34^{k-1}}.\]
\end{thm}
\begin{proof}
Let $\Pi$ be a $\delta$-error protocol for $\AND_k$.
By Lemma~\ref{lemma.mutual2hel} we have that, 
\[I(Z;\Pi(Z))\ge\frac{1}{2^{k-1}\ln2}\cdot\psi(I_j),\]
where $Z\sim\zeta_j$, for any $j\in[k]$, 
Thus, by the definition of $\IC_{\eta,\delta}(\AND_k)$,
\[\IC_{\eta,\delta}(\AND_k)\ge\sum_{j=1}^k
	\frac{1}{k\m32^{k-1}\ln2}\cdot\psi(I_j).\]
Applying in turn Lemmas \ref{lemma.tree}, \ref{lemma.cutandpaste}, and
\ref{lemma.lb} we have
\[\IC_{\eta,\delta}(\AND_k)
	>\frac{\psi(I_{OZ})}{k^2\m32^{k}\ln2}
	=\frac{\psi(I_{EZ})}{k^2\m32^{k}\ln2}
	\ge{\log e\cdot\bigl(1-2\sqrt{\delta(1-\delta)}\bigr)}\cdot\frac{1}{k^2\m34^{k-1}}
,\]
where the application of Lemma~\ref{lemma.lb} is with $A=I_{EZ}$, $t=2^{k-1}$,
$T$ the set of transcripts that output ``1'', and $v$ the all-one vector in $I$.
\end{proof}

It is of interest to note, that
\[\mathrm{IC}_{\eta,\delta}(\AND_k)\le\frac{1}{k}\cdot H(1/2^{k-1})=O(1/2^k).\]
This is achieved by the following protocol. The players, one by one, reveal with
one bit whether they see a 0 or not. The communication ends with the first
player that sees a 0. The amount of information revealed is $H(1/2^{k-1})$
under $\zeta_1$ and $0$ otherwise.

\section{Difficulties in proving a direct-sum theorem}\label{section6}
There seem to be fundamental difficulties in proving a direct-sum theorem on
informational complexity in the NOF model.
The reader familiar with the techniques of 
Bar-Yossef, Jayram, Kumar \& Sivakumar \cite{byjks04}, 
should recall that in the first
part of the method a direct-sum for informational complexity of disjointness is
proved. In particular, it is shown that with respect to suitable collections of
distributions $\eta$ and $\zeta$ for $\DISJ_{n,2}$ and $\AND_2$ respectively,
the information cost of $\DISJ_{n,2}$ is at least $n$ times the informational
complexity of $\AND_2\!:\;\mathrm{IC}_{\eta,\delta}(\DISJ_{n,2})\ge 
n\cdot\mathrm{IC}_{\zeta,\delta}(\AND_2).$
This is achieved via a simulation argument in which the players, to decide the
$\AND_2$ function, use a protocol for disjointness by
substituting their inputs in a special copy of $\AND_2$ and using their random
bits to generate the inputs for the rest $n-1$ copies of $\AND_2$.
In the NOF model the players can no longer perform such a simulation. This is
because, with private random bits, they cannot agree on what the input on the
rest of the copies should be without additional communication. This problem can
be overcome if we think of their
random bits as being not private, but on each player's forehead, just like the
input. 
However, In such a case, although the direct-sum theorem holds, it is useless.
This is because $\mathrm{IC}_{\zeta,\delta}(\AND_k)=0$, as is shown by the
protocol we describe in the next paragraph. 

We describe a protocol that computes $\AND_k$ on every input, with one-sided
error. It has the property that for any distribution over the zeroes of
$\AND_k$, no player learns
anything about his own input. We give the details for three players. Let $x_1$,
$x_2$, $x_3$ denote
the input. Each player has two random bits on his forehead, denoted $a_1$,
$a_2$, $a_3$ and $b_1$, $b_2$, $b_3$. The first player does the following: if
$x_2=x_3=1$, he sends $a_2\oplus a_3$, otherwise he sends $a_2\oplus b_3$. The
other two players behave analogously.  If the {\rm XOR} of the three messages is
`0', they answer `1', otherwise they know that the answer is `0'.  Notice that
any player learns nothing from another player's message. This is because the
one-bit message is {\rm XOR}-ed with one of his own random bits, which he cannot
see.

\bibliographystyle{alpha}
\bibliography{nof}	
\end{document}